\documentclass[conference]{IEEEtran}
\IEEEoverridecommandlockouts

\usepackage{amsmath,amssymb,amsfonts}
\usepackage{graphicx}
\usepackage{xcolor}
\usepackage{array}
\usepackage{cite}

\newif\ifauthornotes
\authornotesfalse

\newtheorem{definition}{Definition}
\newtheorem{lemma}{Lemma}
\newtheorem{proposition}{Proposition}
\newtheorem{theorem}{Theorem}
\newtheorem{corollary}{Corollary}
\newtheorem{assumption}{Assumption}

\newcommand{\bbE}{\mathbb{E}}
\newcommand{\cS}{\mathcal{S}}
\newcommand{\cN}{\mathcal{N}}
\newcommand{\cT}{\mathcal{T}}
\newcommand{\cH}{\mathcal{H}}
\newcommand{\cE}{\mathcal{E}}
\newcommand{\Leaf}{\mathrm{L}}
\newcommand{\TM}{\mathrm{TM}}
\DeclareMathOperator{\Var}{Var}

\begin{document}

\title{Reputation-Graded Wavelet-Packet Subtree\\ Allocation for IoT Sensor Networks}

\author{\IEEEauthorblockN{Surya Jayakumar\IEEEauthorrefmark{1},
Arun Narayanan\IEEEauthorrefmark{2},
Indrakshi Dey\IEEEauthorrefmark{1}}
\IEEEauthorblockA{\IEEEauthorrefmark{1}South East Technological University, Waterford, Ireland}
\IEEEauthorblockA{\IEEEauthorrefmark{2}LUT, Lappeenranta, Finland}
\IEEEauthorblockA{Email: surya.jayakumar@setu.ie, arun.narayanan@lut.fi, indrakshi.dey@setu.ie}}

\maketitle

\begin{abstract}
Access control in Internet of Things (IoT) sensor networks usually acts
above the physical layer: the fusion centre decides which reports to
accept, while the radio of a compromised node keeps occupying the shared
channel. This paper places the entitlement in the waveform: wavelet
packet division multiplexing (WPDM) builds the channel as a binary tree
of orthogonal pulses, and each sensor's reputation sets the depth of the
subtree it may use: trusted sensors hold wide subbands, and low-trust
sensors hold narrow subbands that the fusion centre withdraws one at a
time. We prove that the allocation isolates sensors exactly under
perfect timing, and that under a timing offset the energy leaking out of
the lowpass leaf falls fourfold per tree level. The locally optimum
detector of out-of-subtree transmission in impulsive Class-A noise
weights each leaf in proportion to this predicted leakage. We bound the
time needed to confine a compromised node and derive the optimal tree
depth. In simulation, with demoted nodes radiating within their
allocation, the scheme lowers the probability of a false global decision
by $68\%$ relative to fixed-allocation WPDM and to OFDM with the same
reputation tiers, and by $44\%$ relative to cryptographic
re-attestation, which leads once more than $24\%$ of the sensors are
compromised.
\end{abstract}


\section{Introduction}
\label{sec:intro}

Industrial sensor networks increasingly carry safety-critical decisions.
Groups of low-cost sensors send one-bit local decisions to a decision
fusion centre (DFC) over a shared wireless channel, and the DFC combines
them into a global decision for each monitored event, such as a liquid leak, motor overheating, or the presence of hazardous gas~\cite{dey2022wpdm}. The
channel is shadowed and carries thermal noise and impulsive Class-A
noise from switching equipment~\cite{middleton1977statistical}. Therefore, a
single compromised sensor threatens every group at once: its reports
mislead its own group, and its radio transmits on the subbands of other
sensors and corrupts their honest reports. Access control today acts on the first threat and leaves the second in place \cite{pirayesh2022jamming}, \cite{yao2025}. Upper-layer schemes decide which reports the DFC accepts, while the
compromised radio keeps occupying the channel even after its
credentials are revoked.

Three trends make this problem urgent now. Industrial networks are
adopting zero-trust architectures, which verify every device
continuously~\cite{wan2025continuous}, and these architectures now encompass field sensors whose power budget rules out per-packet cryptography. At the same time, the number of wireless sensors per plant keeps increasing, so that each compromised node shares the channel with many more honest nodes and corrupts more of
them. Finally, sixth-generation (6G) research moves resource allocation
toward agent-based decisions that carry trust
scores~\cite{ghosh2026macro}, which calls for access decisions that the
physical layer itself expresses and verifies.

The state of the art covers parts of this problem. Upper-layer access
control, such as continuous authentication for zero-trust industrial
IoT~\cite{wan2025continuous} and blockchain-based access control with
fast revocation~\cite{ellouze2025lightweight}, grants binary credentials
and checks them through cryptographic exchanges. Physical-layer
authentication identifies a transmitter from channel or hardware
features~\cite{xie2021pla}, and spectrum-misuse detection checks a
permit that the transmitter embeds in its signal~\cite{jin2018specguard};
both verify who transmits and treat the right to a band as binary.
Wavelet Packet Division Multiplexing (WPDM)~\cite{lindsey1997wpm,wong1997wpdm,wong2000wpdm,dey2022wpdm}
divides the band into orthogonal wavelet-packet subbands with high
spectral efficiency and noise immunity, and assumes that every node
transmits where instructed. Reputation-based and Byzantine-resilient
fusion~\cite{ghosh2026macro,su2016byzantine} grade each node based on trust
and acts on its reports once the DFC has received them.

To the best of our knowledge, this paper is the first to tie a graded
trust score to the physical spectrum that a node may use and to verify compliance from signals that the receiver already computes. The WPDM tree makes this possible because it has a hierarchical structure: a node at depth
$j$ holds bandwidth $2^{-j}W$ of the total bandwidth $W$ and sends
symbols of duration $2^{j}T_{0}$, where $T_{0}$ is the root symbol
period. Thus, the tree depth forms a ladder of nested bandwidth–time allocations that can support the implementation of a graded access policy. However, three gaps in the existing literature have so far hindered the development of such a graded access policy. The WPDM literature optimizes the tree for capacity and treats compliance as given, so its correlator bank has served decoding alone. The security literature treats the physical layer as a pipe, so that its trust scores act on data and leave the spectrum fixed. Combining the two requires new analyses: a timing offset spreads a node's energy into neighbouring leaves, and it must be shown that this leakage follows a law set by tree depth and that the optimal detector in impulsive noise weights each leaf according to this leakage. Then, the allocation and the detector become two views of one physical quantity. The advantage is practical: one broadcast of the allocation table withdraws an entitlement, verification adds zero pilots and no sensor-side cost, and the evidence survives the impulsive noise
that degrades the energy detector.

The primary contribution of this paper is the depth-graded allocation itself, in which
reputation selects subtree depth. We prove that it isolates sensors exactly under perfect timing (Lemma~\ref{lem:isolation}) and that, under a timing offset, the leakage of the lowpass leaf falls fourfold per level while every leaf leaks, at most, the slope energy of its pulse (Proposition~\ref{prop:leak}). The second contribution is the detector: we derive the locally optimum test for out-of-subtree transmission in Class-A noise, show that its optimal weights equal the predicted leakage
profile (Proposition~\ref{prop:weights}), and give conditions for a unique optimal tree depth (Theorem~\ref{thm:depth}). The third contribution closes the loop: trimmed-mean reputation limits the bias introduced by colluding reporters to the spread among honest ones (Theorem~\ref{thm:bias}), and the time to confine a compromised node has a closed-form bound (Theorem~\ref{thm:confine}). Simulations against four state-of-the-art alternatives confirm each result.

\section{System Model}
\label{sec:model}

\subsection{Network and wavelet packet tree}

$K$ sensors in $G$ groups send local binary decisions
$u_{k}\in\{-1,+1\}$ to the DFC over a channel of bandwidth $W$. Let
$\phi(t)$ be a unit-energy scaling function whose translates
$\phi(t-nT_{0})$ are orthonormal, with root-mean-square (RMS) bandwidth
$B_{0}$ and lowpass and highpass filters $h[n]$ and $g[n]$. With
$\phi_{0,0}=\phi$ and $T_{j}=2^{j}T_{0}$, the tree grows by the
two-scale recursion~\cite{akansu2000multiresolution}
\begin{equation}
\begin{aligned}
\phi_{j+1,2m}(t)&=\textstyle\sum_{n}h[n]\,\phi_{j,m}(t-nT_{j}),\\
\phi_{j+1,2m+1}(t)&=\textstyle\sum_{n}g[n]\,\phi_{j,m}(t-nT_{j}),
\end{aligned}
\label{eq:basis}
\end{equation}
for $0\le j<J$ and $0\le m<2^{j}$. Each split halves the bandwidth and
doubles the symbol period: node $(j,m)$ holds bandwidth $2^{-j}W$ and
carries symbols on the translates $\phi_{j,m}(t-nT_{j})$. Along the
lowpass branch, the recursion reduces to the dilation
$\phi_{j,0}(t)=2^{-j/2}\phi(2^{-j}t)$. An admissible pruning $\cT$ is a
set of nodes whose subbands tile $[0,W]$; its leaves $\Leaf(\cT)$, with
their translates, form an orthonormal set. The symbol $\ell$ denotes a
leaf $(j,m)$.

\subsection{Time-frequency frame and depth-graded allocation}

Each fusion interval is a frame of $S$ slots of duration $T_{J}$, each
carrying one copy of the tree, so the frame holds $S\,2^{J}$ cells: a
cell is one depth-$J$ node in one slot, of area $WT_{0}$ at every
depth.

\begin{definition}[Depth-graded allocation]
\label{def:graded}
An allocation gives a sensor $k$ a set $\cS_{k}$ of (node, slot) pairs,
disjoint across sensors; sensor $k$ is entitled to transmit on a node in
a slot exactly when the pair lies in $\cS_{k}$. Let $\xi_{k}\in[0,1]$ be
the reputation of sensor $k$ and let $0=\theta_{0}<\dots<\theta_{L}=1$
define $L$ tiers, with tier index
$q(\xi_{k})=\max\{i<L:\theta_{i}\le\xi_{k}\}$. The allocation is
depth-graded when $\cS_{k}$ holds $2^{q}$ nodes at depth $J-q$, with
$J\ge L$ so that every allocated node is a proper subband.
\end{definition}

A node at depth $J-q$ covers $2^{q}$ cells, so a sensor in tier $q$ holds
$A(q)=4^{q}$ cells and loses them one node at a time, in steps of
bandwidth $2^{-(J-q)}W$. In plain terms, a trusted sensor holds several
wide subbands, and a low-trust sensor holds one narrow subband that the
DFC withdraws on its own. New sensors start in tier~$0$. Node sizes are
powers of two, so packing the nodes into the largest slots first leaves
zero gaps, and an allocation is feasible whenever
$\sum_{k}4^{q_{k}}\le S\,2^{J}$. A frame of duration
$T_{F}=K4^{L-1}T_{0}$, that is $S=K4^{L-1}2^{-J}$ slots, admits every
tier assignment at every depth.

\subsection{Received signal and noise}

Sensor $k$ transmits
$s_{k}(t)=\sum_{\ell\in\cS_{k}}a_{k,\ell}\phi_{\ell}(t-\tau_{k})$, where
$a_{k,\ell}$ carries its decision and $\tau_{k}$ is its residual timing
offset. The DFC receives
\begin{equation}
r(t)=\textstyle\sum_{k=1}^{K}h_{k}s_{k}(t)+n(t),
\label{eq:rx}
\end{equation}
with log-normal shadowing gain $h_{k}$, constant over a fusion interval.
The DFC correlates $r(t)$ with every leaf, $y_{\ell}=\langle
r,\phi_{\ell}\rangle$ for all $\ell\in\Leaf(\cT)$, including leaves held
by other sensors and free leaves. This correlator bank is already part
of the WPDM receiver, so the detector of Section~\ref{sec:detect} adds
zero pilot overhead. The noise follows the Middleton Class-A
model~\cite{middleton1977statistical}
\begin{equation}
p(n)=\sum_{i=0}^{Z}\frac{e^{-A}A^{i}}{i!}
\frac{e^{-n^{2}/2\sigma_{i}^{2}}}{\sqrt{2\pi\sigma_{i}^{2}}},
\quad
\sigma_{i}^{2}=\sigma^{2}\frac{i/A+\Gamma}{1+\Gamma},
\label{eq:classa}
\end{equation}
where $A$ is the mean number of noise bursts per symbol, $\Gamma$ the
ratio of Gaussian to impulsive power and $\sigma^{2}$ the total noise
power. We model the leaf-domain noise
$n_{\ell}=\langle n,\phi_{\ell}\rangle$ as independent across leaves
with law~\eqref{eq:classa}, a modelling assumption that keeps the
detector tractable.

\subsection{Isolation and leakage}

\begin{lemma}[Exact isolation]
\label{lem:isolation}
If $\tau_{k}=0$ for all $k$, then
$y_{\ell}=h_{k(\ell)}a_{k(\ell),\ell}+n_{\ell}$ for every leaf $\ell$
with owner $k(\ell)$. Any component of another sensor in $y_{\ell}$
therefore marks a transmission outside that sensor's entitlement.
\end{lemma}
\begin{IEEEproof}
Orthonormality gives
$\langle\phi_{\ell'},\phi_{\ell}\rangle=\delta_{\ell\ell'}$, and
Definition~\ref{def:graded} gives each leaf one owner.
\end{IEEEproof}

\begin{proposition}[Leakage under a timing offset]
\label{prop:leak}
Let $\rho_{\ell,\ell'}(\tau)=\int\phi_{\ell}(t)\phi_{\ell'}(t-\tau)dt$.
The energy leaking out of leaf $\ell'$ into all other leaves and symbol
translates,
$\cE_{\mathrm{leak}}(\ell',\tau)=\sum_{\ell\neq\ell'}|\rho_{\ell,\ell'}(\tau)|^{2}$,
satisfies
$\cE_{\mathrm{leak}}\le\tau^{2}\|\phi_{\ell'}'\|^{2}+O(\tau^{3})$. For
the lowpass leaf $\ell'=(j,0)$,
\begin{equation}
\cE_{\mathrm{leak}}(j,\tau)\le(2\pi B_{0})^{2}\,\tau^{2}\,2^{-2j}+O(\tau^{3}).
\label{eq:leak}
\end{equation}
\end{proposition}
\begin{IEEEproof}
A first-order expansion gives
$\rho_{\ell,\ell'}=\delta_{\ell\ell'}-\tau\langle\phi_{\ell},\phi_{\ell'}'\rangle+O(\tau^{2})$,
and Bessel's inequality bounds the sum by $\|\phi_{\ell'}'\|^{2}$. On the
lowpass branch $\phi_{j,0}'(t)=2^{-3j/2}\phi'(2^{-j}t)$, so
$\|\phi_{j,0}'\|^{2}=2^{-2j}\|\phi'\|^{2}$, and Parseval's theorem gives
$\|\phi'\|^{2}=(2\pi B_{0})^{2}$.
\end{IEEEproof}

\emph{Physical meaning:} A timing offset shifts the pulse, which then
overlaps its neighbours in proportion to its slope. One level deeper, a
lowpass symbol lasts twice as long, the same offset is half as large
relative to it, and the leaked energy falls fourfold. A bandpass leaf
centred at $f_{c}$ has slope energy near
$(2\pi)^{2}(f_{c}^{2}+B_{j}^{2})$, set mainly by $f_{c}$, so its
leakage stays near $(2\pi f_{c}\tau)^{2}$ at every depth. Deeper
allocations thus give finer revocation everywhere and, on the lowpass
branch, higher timing tolerance too.

\subsection{Threat model}

A fraction $f_{B}$ of sensors is compromised. A compromised sensor may report false decisions, transmit on leaves outside its entitlement, or collude to manipulate reputation reports. The DFC is the root of trust and broadcasts the allocation table signed with its key. The analysis uses
three assumptions. (T1)~A compromised sensor radiates within a time-frequency
region proportional to its current allocation, such as when a secure element
sets the transmit mask from the signed table. (T2)~The DFC attributes a
residual to its source from per-node timing and channel estimates.
(T3)~An out-of-subtree transmission carries the residual timing offset
of its link; an attacker that aligns its clock with the DFC removes the
leakage evidence, and the direct residual on the violated leaf then
becomes the evidence. This case is left for future work. Wideband jamming and
DFC compromise lie beyond the scope of this paper.

\section{Reputation Dynamics and Confinement}
\label{sec:rep}

The reputation update follows~\cite{ghosh2026macro}. Sensor $k$ rates
each neighbour $k'\in\cN_{k}$ by exponential smoothing, at rate
$\eta\in(0,1)$, of the consistency $c_{k\to k'}(t)\in[0,1]$ between the
neighbour's report and its own observation:
$\nu_{k\to k'}(t+1)=(1-\eta)\nu_{k\to k'}(t)+\eta c_{k\to k'}(t)$. The
DFC combines the ratings of node $k$ through a trimmed mean
$\TM_{\beta}$, which drops the fraction $\beta$ of highest and lowest
values, and mixes in the physical-layer evidence
$\bar\Lambda_{k}\in[0,1]$ of Section~\ref{sec:detect} with weight
$\omega$:
\begin{align}
U_{k}(t)&=(1-\omega)\TM_{\beta}\big(\{\nu_{k'\to k}(t)\}_{k'\in\cN_{k}}\big)
+\omega\big(1\!-\!\bar\Lambda_{k}(t)\big),
\label{eq:U}\\
\xi_{k}(t+1)&=(1-\eta)\xi_{k}(t)+\eta U_{k}(t).
\label{eq:xi}
\end{align}
The map~\eqref{eq:xi} is a contraction of modulus $1-\eta$, so $\xi_{k}$
tracks $U_{k}$ with a memory of about $1/\eta$ intervals, and
Definition~\ref{def:graded} is applied to $\xi_{k}(t+1)$. On demotion
from tier $q$ to $q-1$, the DFC keeps one child of each of $2^{q-1}$
nodes of $\cS_{k}$, so $A$ falls fourfold, and returns the rest to a
free pool.

\begin{theorem}[Bias under Byzantine reports]
\label{thm:bias}
Let $N=|\cN_{k}|>2\lceil\beta N\rceil$; trim $\lceil\beta N\rceil$
values from each tail; let a fraction $f_{B}<\beta$ of the neighbors
report arbitrary values in $[0,1]$; and let honest reports lie in
$[\bar u-\delta,\bar u+\delta]$. Then
$|\TM_{\beta}-\bar u|\le\delta$, and the fixed point of~\eqref{eq:xi}
satisfies
$|\xi_{k}^{\star}-(1-\omega)\bar u-\omega(1-\bar\Lambda_{k})|\le(1-\omega)\delta$.
\end{theorem}
\begin{IEEEproof}
Each retained value has $\lceil\beta N\rceil$ trimmed values below it
and $\lceil\beta N\rceil$ above it. Fewer than $\beta N$ reports are
adversarial, so an honest value lies on each side, and every retained
value lies in $[\bar u-\delta,\bar u+\delta]$.
\end{IEEEproof}

\emph{Physical meaning:} To shift an average, liars must report extreme
values, and trimming discards exactly those. A minority below $\beta$
therefore moves a node's reputation by at most the natural spread among
honest neighbors, for any number of neighbors.

To bound the confinement time, the DFC latches: once node $k$ is
flagged, it holds $\bar\Lambda_{k}=1$ until $\xi_{k}<\theta_{q-1}$, two
tiers below the tier at detection. This makes the descent deterministic,
but a false alarm holds an honest node near
$(1-\omega)\mathrm{TM}_\beta$, one tier down at $\omega=0.5$.

\begin{theorem}[Confinement time]
\label{thm:confine}
Let node $k$ be compromised at $t=0$ in tier $q\ge2$ with
$\xi_{k}(0)=\xi_{0}$; let each interval give a detection independently
with probability $p$; and let $U_{k}(t)\le u_{1}<\theta_{q-1}$ after the
first detection. Then the confinement time satisfies
\begin{equation}
\bbE[T_{\mathrm{conf}}]\le\frac{1}{p}+
\left\lceil\frac{\ln\big((\theta_{q-1}-u_{1})/(\xi_{0}-u_{1})\big)}
{\ln(1-\eta)}\right\rceil .
\label{eq:confine}
\end{equation}
\end{theorem}
\begin{IEEEproof}
The first detection time is geometric with mean $1/p$. After it,
$U_{k}\le u_{1}$; since~\eqref{eq:xi} increases in both $\xi_{k}$ and
$U_{k}$, induction gives
$\xi_{k}(t)-u_{1}\le(1-\eta)^{t}(\xi_{0}-u_{1})$, and solving
$\xi_{k}(t)\le\theta_{q-1}$ gives the second term.
\end{IEEEproof}

Replacing $\xi_{0}$ by $1$ gives a bound valid for any reputation history
before detection. Under the latch,~\eqref{eq:U} gives
$U_{k}=(1-\omega)\TM_{\beta}$, so $u_{1}$ and the bound depend on
$\omega$. 
\\ 
\emph{Physical meaning:} Confinement is a wait for the first
detection plus a fixed descent set by the reputation memory $1/\eta$. A
node compromised in tier $q$ keeps at most
$4^{q}(\bbE[T_{\mathrm{conf}}]+1)$ entitled cells.

\section{Detection of Out-of-Subtree Transmission}
\label{sec:detect}

\subsection{Test statistic and optimal weights}

Let node $k$ transmit $\phi_{\ell_{k}}$, $\ell_{k}\notin\cS_{k}$, with
gain $c=h_{k}a_{k,\ell_{k}}$; after phase compensation $c$ is real with
arbitrary sign. By Proposition~\ref{prop:leak}, it leaks $c\rho_{\ell}$
into each leaf $\ell$, with $\rho_{\ell}=\rho_{\ell,\ell_{k}}(\tau_{k})$
known from the allocation and the offset estimate. The DFC subtracts the
expected entitled signal, including the leakage that
Proposition~\ref{prop:leak} predicts for every entitled transmission, $z_{\ell}=y_{\ell}-\sum_{\ell'}\hat h_{k(\ell')}\hat a_{k(\ell'),\ell'}\rho_{\ell,\ell'}(\hat\tau_{k(\ell')})$, and tests, with ideal cancellation,
\begin{equation}
\cH_{0}:z_{\ell}=n_{\ell},\qquad
\cH_{1}:z_{\ell}=c\rho_{\ell}+n_{\ell},\qquad \ell\in\cN(\ell_{k}),
\label{eq:hyp}
\end{equation}
where $\cN(\ell_{k})$ holds the leaves $\ell\neq\ell_{k}$ with
significant $\rho_{\ell}$. The violated leaf $\ell_{k}$ also carries its
owner's symbol, and decision-directed cancellation there absorbs part of
the violation, so the detector draws its evidence from the clean
leakage leaves; Section~\ref{sec:results} evaluates (8) with known channel gains and decision-directed
estimates of the entitled symbols, with $P_{FA}$ set per test. Class-A noise is a heavy-tailed mixture, so the likelihood-ratio test
takes a different form from the energy detector. Expanding the
log-likelihood around $c=0$ gives the locally optimum structure: pass
each residual through the score $g(z)=-\frac{d}{dz}\ln p(z)$, then
combine,
\begin{equation}
\Lambda_{k}=\Big(\textstyle\sum_{\ell}w_{\ell}\,g(z_{\ell})\Big)^{2},
\qquad I_{\ell}=\bbE\big[g(z_{\ell})^{2}\,\big|\,\cH_{0}\big],
\label{eq:score}
\end{equation}
and compare $\Lambda_{k}$ with a threshold set for a target false-alarm
rate $P_{FA}$; its normalised value $\bar\Lambda_{k}$ feeds~\eqref{eq:U}.
For Class-A noise, $g$ is a soft limiter~\cite{spaulding1977optimum}: it
grows linearly for small $|z|$ and decays for large $|z|$.

\begin{proposition}[Deflection-optimal weights]
\label{prop:weights}
Let $S=\sum_{\ell}w_{\ell}g(z_{\ell})$ and
$d^{2}(w)=(\bbE[S|\cH_{1}]-\bbE[S|\cH_{0}])^{2}/\Var(S|\cH_{0})$. For
small $|c|$, $d^{2}$ is maximised by $w_{\ell}^{\star}\propto\rho_{\ell}$,
with
\begin{equation}
d^{2}_{\max}=c^{2}\textstyle\sum_{\ell\in\cN(\ell_{k})}\rho_{\ell}^{2}I_{\ell}.
\label{eq:dmax}
\end{equation}
\end{proposition}
\begin{IEEEproof}
Under $\cH_{0}$ the score has zero mean and
$\Var(S)=\sum_{\ell}w_{\ell}^{2}I_{\ell}$. For small $c$,
$\bbE[g(z_{\ell})|\cH_{1}]=c\rho_{\ell}I_{\ell}+o(c)$, so
$d^{2}=c^{2}(\sum w_{\ell}\rho_{\ell}I_{\ell})^{2}/\sum w_{\ell}^{2}I_{\ell}$,
and the Cauchy-Schwarz inequality gives $w^{\star}\propto\rho$.
\end{IEEEproof}

The isolation analysis and the detector thus describe one physical
quantity. With finite third moments of
$g(z_{\ell})$ and a vanishing share of $\sum w_{\ell}^{2}I_{\ell}$ in
any single leaf, the Lyapunov condition holds and $S$ is asymptotically
Gaussian; with $t=Q^{-1}(P_{FA}/2)$, $P_{D}=Q(t-d)+Q(t+d)$. For
log-normal gains we evaluate this law with the exact mean and variance
of $g(c\rho_{\ell}+n_{\ell})$.

\subsection{Optimal tree depth}

A deeper tree enlarges the combining set, while the energy per leaf
falls; the correlator bank, the allocation table and the pool management
grow with the number of leaves. Write $M(J)=|\cN(\ell_{k})|$, $\gamma(J)$
for the per-leaf signal-to-noise ratio and $\psi(\gamma)$ for the
per-leaf deflection, so $d(J)=\sqrt{M(J)}\,\psi(\gamma(J))$, and let
$C(J)$ be the per-interval overhead at price $\mu>0$. The design
objective is $F(J)=\ln d(J)-\mu C(J)$.

\begin{assumption}
\label{asm:depth}
(A1) $\ln M(J)$ is concave and increasing; (A2) $\ln\psi(\gamma(J))$ is
concave in $J$; (A3) $C(J)$ is increasing and strictly convex on
$[J_{\min},J_{\max}]$, with $J_{\min}=L$.
\end{assumption}

\begin{theorem}[Unique optimal depth]
\label{thm:depth}
Under Assumption~\ref{asm:depth}, $F(J)$ is strictly concave on
$[J_{\min},J_{\max}]$ with a unique maximiser $J^{\star}$, which solves $F'(J^{\star})=0$ when interior and sits at an endpoint in the remaining cases. For integer
depth, the optimum is the better of $\lfloor J^{\star}\rfloor$ and
$\lceil J^{\star}\rceil$.
\end{theorem}
\begin{IEEEproof}
$\ln d=\frac12\ln M+\ln\psi(\gamma)$ is concave by (A1)-(A2), $-\mu C$
is strictly concave by (A3), and a strictly concave function on a
compact interval has exactly one maximiser.
\end{IEEEproof}

\begin{corollary}[Geometric regime]
\label{cor:geometric}
If $M=M_{0}2^{J}$, $\gamma=\gamma_{0}2^{-J}$, $\psi=\kappa\gamma^{a}$ and
$C=C_{0}2^{J}$, then $\ln d$ has slope $(\frac12-a)\ln2$. For
$a<\frac12$, $J^{\star}=\log_{2}[(\frac12-a)/(\mu C_{0})]$ clipped to
$[J_{\min},J_{\max}]$; for $a\ge\frac12$, $J^{\star}=J_{\min}$.
\end{corollary}

\emph{Physical meaning:} Splitting the tree spreads the evidence over
more leaves. When the evidence per leaf falls faster than the number of
leaves grows ($a\ge\frac12$), the DFC sees violations best in a shallow
tree, and the overhead term strengthens that preference. More generally,
whenever $d(J)$ decreases and $C(J)$ increases, $F(J)$ decreases and
$J^{\star}=J_{\min}$, with or without (A2).

\section{Simulation Setup}
\label{sec:sim}

\begin{table}[!t]
\centering
\caption{Simulation parameters}
\label{tab:params}
\footnotesize
\setlength{\tabcolsep}{3pt}
\begin{tabular}{lll}
\hline
Parameter & Symbol & Value\\
\hline
Sensors, groups & $K,G$ & 64, 4\\
Tree depth, tiers & $J,L$ & 4 (swept 3-7), 3\\
Tier thresholds & $\theta_{1},\theta_{2}$ & 1/3, 2/3\\
Frame length & $T_{F}$ & $1024\,T_{0}$ ($S=64$ at $J=4$)\\
Compromised fraction & $f_{B}$ & 0-0.30, default 0.15\\
Trim, smoothing, evidence wt. & $\beta,\eta,\omega$ & 0.30, 0.20, 0.5\\
Timing offset & $\tau/T_{0}$ & 0.1 (swept to 0.2)\\
Class-A parameters & $A,\Gamma,Z$ & 0.1, 0.1, 4\\
Shadowing, false-alarm rate & $\sigma_{s},P_{FA}$ & 8\,dB, 0.01\\
Local accuracy, re-attest.\ period & $p_{1},T_{a}$ & 0.85, 20\\
Scaling function & $\phi$ & db8 (db4, db16, sym8, coif4)\\
\hline
\end{tabular}
\end{table}

Table~\ref{tab:params} lists the configuration; $\beta$, $\eta$ and
$f_{B}$ follow~\cite{ghosh2026macro}. Each group fuses the decoded
reports of its active members by majority vote. A compromised node
reports the wrong hypothesis and, in each interval, transmits on a
random leaf of another node; under (T1) the transmission hits with
probability $A(q)/A(q_{0})=4^{q-q_{0}}$ after demotion from its initial
tier $q_{0}$, and with probability one under fixed allocation. The
fusion model depends on the allocation through this ratio alone, so its
results hold for every feasible frame. Detector
outputs follow the $\cH_{0}$ and $\cH_{1}$ laws of
Section~\ref{sec:detect} at $\tau=0.1T_{0}$. Each point averages 40 independent runs of 150 fusion intervals across four groups, yielding 24,000 group decisions.

\emph{State-of-the-art baselines.} Each baseline represents one class of
the state of the art reviewed in Section~\ref{sec:intro}. All schemes use the same frame and the same cell
budget per tier, so they operate at equal spectral efficiency: (i)~WPDM with
cryptographic re-attestation every $T_{a}$ intervals, the upper-layer
access control of~\cite{wan2025continuous,ellouze2025lightweight}, which
revokes the credentials of compromised nodes and keeps their
allocation; (ii)~fixed-allocation WPDM~\cite{dey2022wpdm} with
reputation filtering of reports; (iii)~OFDM subcarriers with the same tiers and parameters,
which isolate the value of the tree hierarchy; its cyclic prefix,
longer than $\tau$, removes inter-carrier leakage, so its evidence term
carries false alarms alone, and, as in WPDM, the direct residual is left
out;
and (iv)~reputation-based fusion
alone~\cite{ghosh2026macro,su2016byzantine}. The detector baselines are the energy detector and the locally optimum detector with equal weights, \(w_\ell=1\) for all \(\ell\).

\section{Numerical Results}
\label{sec:results}

\begin{figure}[!t]
\centering
\includegraphics[width=\columnwidth]{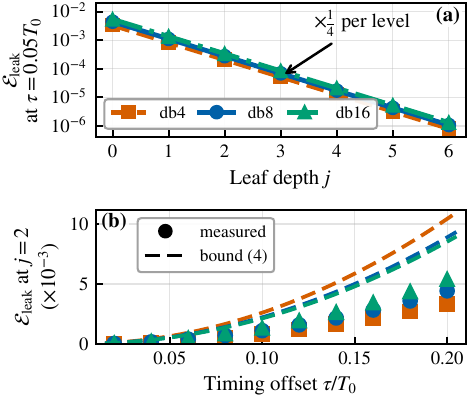}
\caption{Energy leaking out of the lowpass leaf under a timing offset.
(a)~Against leaf depth $j$ at $\tau=0.05T_{0}$; sym8 and coif4 coincide
with db8. (b)~Against $\tau$ at $j=2$; symbols are simulation, dashed
lines the bound~\eqref{eq:leak} for each family.}
\label{fig:leakage}
\end{figure}

\subsubsection{Leakage under a timing offset}

Fig.~\ref{fig:leakage}(a) shows that the lowpass-leaf leakage falls by a
factor of 4.00 per level (log slope $-1.386$), the $2^{-2j}$ law
of~\eqref{eq:leak}: the same clock error is a smaller fraction of a
longer symbol. Fig.~\ref{fig:leakage}(b) confirms growth with $\tau^{2}$
(exponent 2.000), with the bound above every measurement by a factor of
1.6 (db16) to 3.1 (db4); $(2\pi B_{0})^{2}$ ranges from 3.41 (db16) to
4.17 (db4) and is 3.554 for both db8 and sym8. The markers count only
leakage into the same subband; counting every leaf of the db8 tree, the
leakage matches~\eqref{eq:leak} within $1\%$, so the gap is the energy
that crosses into other subbands, largest for the short db4 filter.

\begin{figure}[!t]
\centering
\includegraphics[width=\columnwidth]{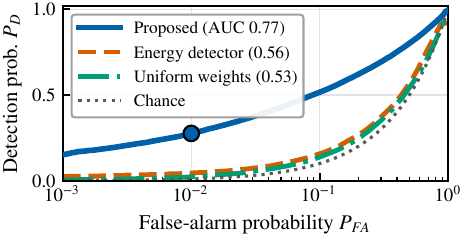}
\caption{ROC of out-of-subtree detection in Class-A noise ($J=4$,
$\tau=0.1T_{0}$). The marker shows the operating point $P_{FA}=0.01$,
$P_{D}=0.276$.}
\label{fig:roc}
\end{figure}

\begin{figure}[!t]
\centering
\includegraphics[width=\columnwidth]{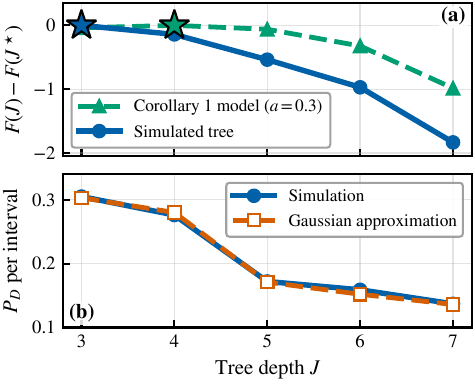}
\caption{Choice of tree depth. (a)~Objective $F(J)$ relative to its
maximum; for the simulated tree, $d(J)$ is the effective deflection
behind panel~(b). Stars mark the maximisers. (b)~Detection probability
per interval at $P_{FA}=0.01$ and $\tau=0.1T_{0}$.}
\label{fig:depth}
\end{figure}

\subsubsection{Detection and tree depth}

Fig.~\ref{fig:roc} compares the detectors. The proposed detector
reaches an area under the ROC curve (AUC) of 0.77, against 0.56 for the
energy detector and 0.53 for uniform weights. At $P_{FA}=0.01$, it
detects a violation with probability 0.276 per interval, and the
Gaussian approximation predicts 0.28. Class-A bursts are rare and strong. The energy
detector counts each burst as evidence and must raise its threshold
above the burst level, which also hides the weak leakage. The soft
limiter clips the bursts, and the weights $w^{\star}\propto\rho$ act as a
matched filter for the leakage pattern. Uniform weights add leaves where
the leakage changes sign, and so the contributions cancel.

Fig.~\ref{fig:depth} examines Theorem~\ref{thm:depth} over the
admissible depths $J\ge L=3$, with $C(J)=C_{0}2^{J}$ and
$\mu C_{0}=0.0125$ for both curves. In the geometric regime of
Corollary~\ref{cor:geometric} with $a=0.3$, the closed-form optimum
$J^{\star}=4$ matches the numerical maximiser; the overhead
constant of this curve is chosen to give $J^{\star}=4$; thus, the curve
illustrates the corollary. For the simulated tree, $d(J)$ is the effective
deflection that yields the simulated $P_{D}$ of Fig.~\ref{fig:depth}(b)
under the Gaussian law of Section~\ref{sec:detect}-A: it falls at every
step, from 2.07 at $J=3$ to 1.48 at $J=7$, therefore, $F(J)$ decreases and its
maximiser is the shallowest admissible depth, $J^{\star}=J_{\min}=3$.
$P_{D}$ itself falls from 0.305 at $J=3$ to 0.136 at $J=7$, and the
Gaussian approximation fits it at every depth. We use $J=4$: relative to $J=3$, it halves the bandwidth
of a tier-0 subband, the finest revocation step, for a loss of 0.029
in $P_{D}$. A deeper tree cuts the violator's leakage into
more, weaker pieces, and each piece sinks into the noise faster than the
number of pieces grows. Detection favours shallow trees, fine revocation
favours deep trees, and $J=4$ balances the two.

\begin{figure}[!t]
\centering
\includegraphics[width=\columnwidth]{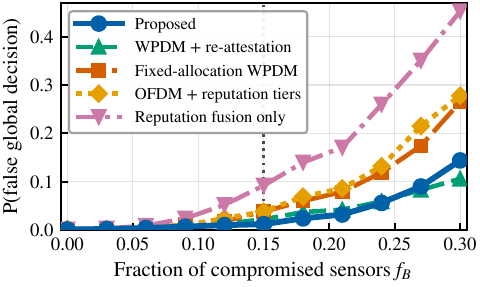}
\caption{Probability of a false global decision against the fraction of
compromised sensors for the proposed scheme and the four
state-of-the-art baselines at equal spectral efficiency. Dotted line:
default $f_{B}=0.15$.}
\label{fig:sota}
\end{figure}

\subsubsection{Comparison with the state of the art}

Fig.~\ref{fig:sota} compares the proposed scheme with the baselines of
Section~\ref{sec:sim}. At the default $f_{B}=0.15$, the probability of a
false global decision is 0.012 for the proposed scheme, 0.022 for
re-attestation, 0.039 for fixed-allocation WPDM and for OFDM tiers, and
0.092 for reputation fusion alone: reductions of $44\%$, $68\%$ and
$87\%$. The proposed scheme also corrupts $83\%$ fewer leaves per
interval than fixed allocation and about one sixth of the leaves
corrupted under re-attestation. Re-attestation leads above
$f_{B}\approx0.24$, and all schemes perform alike for $f_{B}<0.06$. Demotion shrinks the band a compromised node
holds, so under (T1) its stray transmissions hit fewer honest leaves;
re-attestation removes the node's reports and leaves its band in place.
OFDM subcarriers sit side by side at one level, so the revocation step
stays fixed as trust changes, and the cyclic prefix removes the leakage
evidence. Once a large share of sensors is compromised, the trimmed
mean loses its honest majority and cryptographic revocation becomes the
stronger tool: the two mechanisms complement each other.

\begin{figure}[!t]
\centering
\includegraphics[width=\columnwidth]{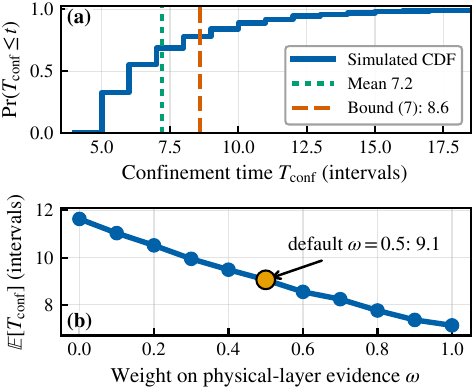}
\caption{Confinement of a node compromised in tier~2. (a)~Empirical CDF
of the confinement time at $\omega=1$, with its mean and the
bound~\eqref{eq:confine}.
(b)~Mean confinement time against the evidence weight $\omega$
of~\eqref{eq:U}.}
\label{fig:confine}
\end{figure}

\subsubsection{Confinement and reputation robustness}

Fig.~\ref{fig:confine}(a) shows the
confinement time with the evidence weight $\omega=1$. The latch then
gives $u_1=0$, and (7) with $\xi_0=1$, $\theta_1=1/3$, $p=0.276$ and
$\eta=0.2$ gives $1/p+\lceil\ln(1/3)/\ln 0.8\rceil=3.6+5=8.6$ intervals.
The simulated mean is 7.2 intervals, below this bound, and every node is
confined after at least 5 intervals, exactly the descent term; the tail
comes from late first detections. Fig.~5(b) shows the mean confinement
time falling steadily from 11.6 intervals at $\omega=0$, where the
neighbour ratings act alone, to 7.1 at $\omega=1$, with 9.1 at the
default $\omega=0.5$, where the latch holds $U_k$ at
$0.5\,\mathrm{TM}_\beta>0$ and the descent takes longer: the
physical-layer evidence shortens confinement by 22\% at $\omega=0.5$
and by 39\% at $\omega=1$. At $f_B=0.15$, false alarms keep honest nodes
confined for $0.6\%$ of intervals and one tier down for $28\%$; this
costs honest bandwidth but not decision reliability, since tier-1
reports are still fused. With $N=20$ and $\delta=0.1$, the trimmed-mean
deviation stays at or below $0.083<\delta$ for every $f_B<\beta$, as
Theorem~1 predicts, and exceeds $\delta$ from $f_B=0.35$. The correlator
bank sees a violation directly; neighbour ratings see it later, through
the reported decisions. The default $\omega=0.5$ keeps the neighbour
ratings, and with them the protection of Theorem~1, as a second source
of evidence for attackers that align their timing (T3). The detector
runs at the DFC, at about 80 exponentials per test over the median
$|\mathcal{N}(\ell_k)|=16$ leaves, and adds zero cost at the sensor;
re-attestation asks every sensor to sign every $T_a$ intervals.

\section{Conclusion}
\label{sec:concl}

Tree depth in WPDM works as a graded access primitive: the waveform
carries the entitlement, and the correlator bank holds the evidence of a
violation. On the lowpass branch, leakage falls fourfold per level, and the optimal
detector, which weights each leaf by this predicted leakage, reaches an
AUC of 0.77 in Class-A noise, against 0.56 for the energy detector.
Detection evidence weakens with depth, so the optimal depth is the
shallowest admissible one, $J=3$, and $J=4$ trades a small detection
loss for finer revocation. At the default compromised fraction, the scheme
reduces false global decisions by $68\%$ relative to fixed-allocation
WPDM and OFDM tiers and by $44\%$ relative to re-attestation. These
results hold under the emission model (T1) and the leaf-domain noise
model of Section~\ref{sec:model}. Future work will focus on timing-aligned
attackers through the direct residual, attribution from per-node
signatures, a measured overhead model, and distributed fusion.


\section*{Acknowledgment}

\noindent This work is supported by the EU Marie Skłodowska-Curie Actions Project “COALESCE” under Grant Number 101130739 and US-Ireland R\&D Partnership Programme Project ``Resilient Networks'' under Grant RI-SFI-23/US/3924.

\end{document}